\documentclass[conference]{IEEEtran}
\IEEEoverridecommandlockouts

\usepackage[utf8]{inputenc}
\usepackage{amsthm}
\usepackage{amssymb}
\usepackage{graphicx,setspace}
\usepackage{amsmath,amssymb}
\usepackage{color}
\usepackage[english]{babel}
\usepackage{mathrsfs,dsfont}
\usepackage{array}
\usepackage{tikz}
\usepackage{adjustbox}
\usepackage{multirow}
\usepackage{float}
\usepackage{hyperref}
\usepackage{booktabs}
\usepackage{graphicx}
\usepackage{lineno}
\usepackage{rotating}
\usepackage{pdflscape}
\usepackage{dsfont}
\usepackage{diagbox}
\usepackage{tikz}
\usepackage{tikz-network}
\usepackage{setspace}
\usepackage{mathrsfs}
\usepackage[mathscr]{euscript}
\usepackage{bbm}

\usepackage{mathtools}

\usepackage[font=scriptsize,labelfont=bf]{caption}
\usepackage{inputenc}

\newtheorem{theorem}{\bf Theorem}
\newtheorem{lemma}[theorem]{\bf Lemma}

\newtheorem{example}{Example}
\newtheorem{assumption}{Assumption}

\newtheorem{definition}{Definition}

\newcommand{\nn}{\nonumber}

\def\QED{~\rule[-1pt]{5pt}{5pt}\par\medskip}
\renewenvironment{proof}{{\bf Proof: \ }}{\hfill \QED}

\let\ALP  \mathcal

\newcommand{\beqq}[1]{\begin{align*} #1 \end{align*}}

\newcommand{\Na}{\mathbb{N}}

\DeclareUnicodeCharacter{03BB}{$\lambda$}

\allowdisplaybreaks

\title{Change Detection of Markov Kernels with Unknown Pre and Post Change Kernel}

\author{Hao Chen \and
        Jiacheng Tang \and
        Abhishek Gupta\thanks{Jiacheng Tang is with GE Global Research, Niskayuna, NY. Hao Chen and Abhishek Gupta are with the Department of Electrical and Computer Engineering at The Ohio State University, Columbus, OH, USA. Email: {\tt\small jiacheng.tang@ge.com, chen.6945, gupta.706@osu.edu}.}}

\begin{document}
\maketitle
\begin{abstract}
In this paper, we develop a new change detection algorithm for detecting a change in the Markov kernel over a metric space in which the post-change kernel is unknown. Under the assumption that the pre- and post-change Markov kernel is uniformly ergodic, we derive an upper bound on the mean delay and a lower bound on the mean time between false alarms. A numerical simulation is provided to demonstrate the effectiveness of our method. 
\end{abstract}
\section{Introduction}
Change detection algorithms take a stream of observations as input and detect whether or not the statistical properties of an underlying system (generating observations) have changed. These algorithms form the bedrock of fault and attack detection in dynamic systems \cite{tang2021dynamic} and anomaly detection in signal processing and industrial systems \cite{ranshous2015anomaly, veeravalli2014quickest}. We propose a new kernel-based online change detection algorithm for Markov chains that is model-free (prior information about the transition kernels is not needed) and derive its performance bounds. The two-sample test is at the core of our change detector, as its main purpose is to determine whether two groups of samples come from the same distribution. A change detector, in essence, does this sequentially for data streams. The kernel-based two-sample test is developed in \cite{gretton2012kernel}, which is advantageous in data-driven applications, especially when the sample space is high dimensional. It embeds probability distributions of different groups of samples into a reproducing kernel Hilbert space (RKHS) through an injective map and measures the distance between the embeddings using the RKHS norm. The metric between the probability measures induced by this distance is called the maximum mean discrepancy (MMD). The difference in distributions is detected when the MMD exceeds a certain threshold.

The theory of RKHS embeddings of probability measures can be found in \cite{sriperumbudur2010hilbert}, which provides inspiration and a theoretical foundation for our work. MMD-based change detection has received significant attention in the past decade \cite{Li2015mstatistic} \cite{flynn2019change} \cite{li2019scan} \cite{zaremba2013b}. The test statistics of these methods are built around the empirical estimation of the MMD. Two major online change detection frameworks are the cumulative sum (CUSUM) type and (scanning) window type. The kernel CUSUM method introduced in \cite{flynn2019change} is an adaptation of the well-established CUSUM statistics \cite{page1954continuous} by substituting the distance function from the empirical KL-divergence to the empirical MMD.

The performance guarantees for the change detection methods are given in the form of bounds on the mean time between false alarms (MTBFA), also known as the average run length (ARL), and the mean delay (MD). When samples are i.i.d., the upper bound for MTBFA and lower bound for MD are obtained using standard concentration inequalities for supermartingales in \cite{flynn2019change}. The window type approach presented in \cite{li2019scan} employs a specially designed sliding window framework to reduce computation complexity during the estimation of MMD. The performance bounds are derived using a sophisticated change-of-measure technique. However, for samples generated from aperiodic Markov chains, the analysis in the aforementioned studies breaks down due to the dependency structure, which motivates us to search for a new set of tools. 

Change detection for Markov chains has been studied earlier \cite{xian2016online}. The authors adapt the CUSUM statistics for uniformly ergodic Markov chains over finite state spaces. The upper bound on MTBFA is obtained via characterizing the distribution of the test statistic as time goes to infinity, and the lower bound on MD is derived by applying a version of Hoeffding's inequality \cite{glynn2002hoeffding} for uniformly ergodic Markov chains. In contrast, our setting is in the general state space, but the uniform ergodicity condition is kept, as Hoeffding's inequality is a crucial step in our performance analysis.

% Another study that deals with both the kernel method and ergodic dynamic systems is \cite{Solowjow2020dynamic}. The difference between ours and theirs is two-fold: 1) their detection scheme is offline, that is, gathering samples first and testing after; 2) the dependence between samples is dealt with indirectly through picking a sufficiently large sampling interval to ensure the $\beta$-mixing coefficients are small for samples used in the test statistic.

\textbf{Our Contribution:} Inspired by the existing approaches, we devised an online change detection algorithm for uniformly ergodic Markov Chains with general state space. The primary contributions of this paper are summarized below:

\begin{itemize}
    \item The kernel-based \textbf{online} change detection algorithm we propose works with \textbf{weakly dependent samples} generated from uniformly ergodic Markov chains. This is different from \cite{Solowjow2020dynamic}, which avoids the dependency through downsampling of the system trajectory.
    
    \item The proposed algorithm addresses the change detection problem with both \textbf{unknown pre- and post-change} Markov kernels. This is an improvement over the method in \cite{xian2016online} which requires the knowledge of the pre-change transition matrix. 
    
    \item The Markov chains are assumed to have \textbf{general state space}, i.e., metric space, which makes our method superior to existing methods that assume finite state space \cite{xian2016online}.
    
    \item Our algorithm detects the change in product measure $\pi_P \otimes P$. Testing for the invariant distribution for the Markov chain is insufficient, as one can construct Markov kernels $Q \neq P$ with the invariant measure $\pi_Q = \pi_P$.
\end{itemize}

The rest of the paper is organized as follows. In Section \ref{sec: preliminaries}, we review the definition and preliminary results of MMD. In Section \ref{sec: problem formulation}, we formally state our problem setup. We present the test statistics and the main results deriving the bounds on MTBFA and MD in Section \ref{sec: main results}, followed by the numerical simulations in Section \ref{sec: numerical simulations}.

\section{Preliminaries}
\label{sec: preliminaries}
In this section, we introduce the notion of RKHS, kernel mean embedding, maximum mean discrepancy (MMD), and weak dependence. The key reference for the used terminology is \cite{sriperumbudur2010hilbert}. Consider a metric space $(\mathbb{X}, \mathscr{X})$ and a positive definite kernel function $k: \mathbb{X}\times\mathbb{X} \to \mathbb{R}$ (see \cite{sriperumbudur2010hilbert} for definition of positive definite kernels). There exists a reproducing kernel Hilbert space (RKHS) $\mathcal{H}_k$ of functions equipped with an inner product $\langle\cdot, \cdot \rangle_{\mathcal{H}_k}$. For any $f\in \mathcal{H}_k$, $f(x) = \langle f, k(x, \cdot)\rangle_{\mathcal{H}_k}$ due to the reproducing property. 

We assume the kernel is bounded, and without loss of generality, we set the bound to be 1, i.e., $\sup_{x,y\in\mathbb{X}} |k(x,y)|\leq 1$. Let $\mathsf{P}$ be a measure in the space of probability measures over $\mathbb X$, denoted by $\mathcal{P}(\mathbb X)$. The \textit{kernel mean embedding} $\mu_\mathsf{P}$ is defined as $\mu_\mathsf{P}(\cdot) \coloneqq \mathbb{E}_{X\sim\mathsf{P}}[k(X, \cdot)]$. This mapping from $\mathcal{P}$ to $\mathcal{H}_k$ is injective if and only if the kernel $k$ is characteristic.

The integral probability metric (IPM) between two measures $\mathsf{P}, \mathsf{Q}\in\mathcal P(\mathbb X)$ is defined as
$$\gamma = \sup_{f\in\mathcal{F}}\left | \mathbb{E}_{X\sim\mathsf{P}}[f(X)] - \mathbb{E}_{X\sim\mathsf{Q}}[f(X)] \right |,$$
where $\mathcal{F}$ is some class of functions. It is referred to as maximum mean discrepancy (MMD) when $\mathcal F$ is chosen as the unit ball in the RKHS $\mathcal{H}_k$. Under the assumption that the kernel function $k$ is bounded, the MMD can also be written as the distance between the embedding $\gamma_k(\mathsf{P}, \mathsf{Q}) = \|\mu_\mathsf{P}-\mu_\mathsf{Q}\|_{\mathcal{H}_k}$.
Note that $\gamma_k$ is a pseudometric in when a general selection of the kernel function, since it fails to satisfy $\gamma_k(\mathsf{P}, \mathsf{Q}) = 0 \implies \mathsf{P} = \mathsf{Q}$. The kernel $k$ being characteristic is a necessary and sufficient condition for $\gamma_k$ to be a metric. The square MMD is more commonly used in statistical hypothesis testing:
\begin{align}
\label{eqn: squared mmd}
    \gamma^2_k(\mathsf{P}, \mathsf{Q})
    & = \mathbb{E}_{X, X'\sim\mathsf{P}}[k(X, X')] - 2\mathbb{E}_{X\sim\mathsf{P}, Y\sim\mathsf{Q}}[k(X, Y)] \nonumber\\
    & + \mathbb{E}_{Y,Y'\sim\mathsf{Q}}[k(Y, Y')].
\end{align}
Several criteria for characteristic kernels are provided in \cite{sriperumbudur2010hilbert}. When the kernel is translation invariant, i.e. $k(x, y) = \psi(x-y)$ and defined on $\mathbb{R}^d$, $k$ is characteristic if and only if the support of its Fourier transform is $\mathbb{R}^d$ \cite{sriperumbudur2010hilbert}. For a comprehensive characterization of kernels $k$ such that $\gamma_k$ is a metric, please refer to Table 1 in \cite{muandet2016kernel}. The Gaussian kernel is frequently used as the characteristic kernel for MMD:
\begin{align}
\label{ex: Gaussian kernel}
    k(x,y) = \psi(x-y) &= \exp{\left(-\frac{\|x-y\|^2}{2\sigma^2}\right)}, x, y \in \mathbb{R}^d, \sigma > 0.
\end{align}

\section{Problem Formulation}
\label{sec: problem formulation}

Consider a Markov chain $\{X_i\}_{i\in\mathbb{N}}$ defined on a probability space $(\Omega, \mathcal{F}, \mathbb{P})$ with values in a metric space $(\mathbb{X}, \mathscr{X})$. Let $P: \mathbb{X} \to \mathcal{P}(\mathbb{X})$ denote the transition kernel of the Markov chain. At some unknown time $\tau \in\mathbb{N}$, the Markov chain immediately switches its transition kernel to $Q$ ($P \neq Q$). Both the pre-change kernel $P$ and the post-change kernel $Q$ are assumed to be uniformly ergodic, and they satisfy Doeblin's condition, which is defined below.
\begin{assumption}
    [Doeblin's Condition]
    \label{asp: Doeblin condition}
        A Markov kernel $P$ satisfies Doeblin condition if there exist a probability measure $\phi$ on $\mathbb{X}$, $\lambda > 0$, and some integer $l\geq 1$ such that
        $P^l(x, \cdot) \geq \lambda\phi(\cdot),$
        for all $x\in\mathbb{X}$.
        % , which is referred to as Doeblin's condition.
\end{assumption}
%explained in Assumption \ref{thm: Doeblin condition}. 
% aperiodic recurrent Harris chain and satisfies 

The Doeblin's condition is equivalent to uniform ergodicity for aperiodic irreducible Markov chains. We refer the reader to Theorem 16.2.3 in \cite{meyn2012markov}. Note that, Doeblin's coefficients for $P$ and $Q$ do not need to be the same. We need Doeblin's condition to hold for the Markov chains discussed in this paper as it allows us to apply Hoeffding's inequality, which is key to obtaining the performance bounds. Although Assumption \ref{asp: Doeblin condition} seems rather restrictive from a Markov chain perspective, it is satisfied by many dynamic systems (linear and nonlinear systems). In Section \ref{sec: example section}, we discuss the characterization of such systems satisfying \ref{asp: Doeblin condition}.

To properly define a distribution of samples generated from the pre- and post-change Markov chains, we now introduce the general RKHS embedding of Markov kernels.

\subsection{RKHS Embedding of Markov Kernels}
Consider a measurable space $(\mathbb{X}, \mathscr{X})$ and let $P: \mathbb{X} \to \mathcal{P}(\mathbb{X})$ be a Markov kernel. Assume that the kernel $P$ admits an invariant probability measure $\pi: \mathscr{X} \to [0, 1]$. Define a probability measure $F_P$ on measurable space $(\mathbb{X}\times\mathbb{X}, \mathscr{X}\otimes\mathscr{X}) \coloneqq
(\mathbb{Z}, \mathscr{Z})$ as  
\begin{equation}
\label{eqn: second order kernel}
    F_P(A\times B) = \pi \otimes P (A \times B) = \int_{A}\pi(dx)P(x, B), 
\end{equation}
$\text{for } A,B\in\mathscr{X}$, where $\otimes$ denotes the product $\sigma$-algebra or product measure. Let $k: \mathbb{Z}\times \mathbb{Z} \to \mathbb{R}$ be a characteristic kernel function, and $\ALP H_k$ be the corresponding RKHS. Then the RKHS embedding of $F_P$ is written as $ \mu_P(\cdot) = \mathbb{E}_{z \sim F_P}(k(z, \cdot))$. In the following lemma, we identify a condition under which $P\neq Q$ implies $F_P\neq F_Q$.
\begin{lemma}
\label{lem: RKHS for MC}
    Let $P$ and $Q$ be two Markov kernels with invariant distributions $\pi_P$ and $\pi_Q$, respectively. If there exists a closed set $\mathbb A$ in $\ALP X$, such that $P(x,\cdot)\neq Q(x,\cdot)$ for all $x\in\mathbb A$ and $\pi_P(\mathbb A)>0$, then $F_P\neq F_Q$.
\end{lemma}
\begin{proof}
    Please refer to Appendix \ref{app: RKHS for MC}.
    % The proof is straightforward and is therefore omitted. 
\end{proof}
Denote $\{\tilde X_i\}_{i\in\mathbb{N}} = \{(X_i, X_{i+1})\}_{i\in\mathbb{N}}$ as the second-order Markov chain for the original Markov chain $(X_i)_{i\in\mathbb{N}}$. In the following lemma, we show that the second-order Markov chain also satisfies Doeblin's condition.

\begin{lemma} \label{lem: UE for extended space}
    Let $\{X_i\}_{i\in\mathbb{N}}$ be a Markov chain satisfying Doeblin's condition, and $\{\tilde X_i\}_{i\in\mathbb{N}} = \{(X_i, X_{i+1})\}_{i\in\mathbb{N}}$ be the second-order Markov chain. Then, $\{\tilde X_i\}_{i\in\mathbb{N}}$ also satisfies Doeblin's condition.
\end{lemma}

\begin{proof}
    Please refer to Appendix \ref{app: UE for extended space}.
\end{proof}

Let $F_P$ and $F_Q$ be defined as in \eqref{eqn: second order kernel} for $P$ and $Q$, respectively. The sequential hypothesis test at time $t$ is
\begin{align*}
    \ALP H_0: \tilde X_i \sim F_P,\, \forall i<t; \quad
    \ALP H_1: 
    \tilde X_i \sim
    \begin{cases}
        F_Q,\quad i > \tau,\\
        F_P, \quad \text{otherwise,}
    \end{cases}
\end{align*}
$\exists 1\leq \tau < t$. The stopping time $T$ can be interpreted as the first time the null hypothesis $\ALP H_0$ is rejected.
To evaluate the detection performance of $T$, consider the mean time between false alarms (MTFBA) and mean delay (MD) as follows.
\begin{align}
\label{def: MTBFA+MD}
    \text{MTBFA}(T) &= \mathbb{E}_\infty[T]\\ 
    \text{MD}(T) &= \mathbb{E}_{\tau}[T-\tau|T>\tau],
\end{align}
where $\mathbb{E}_\infty$ denotes the expectation under $\mathcal{H}_0$, i.e., $\tau=\infty$, and $\mathbb{E}_\tau$ denotes the expectation under $\mathcal{H}_1$, i.e., $\tau$ is finite. 

To distinguish $F_P$ and $F_Q$, we adopt maximum mean discrepancy (MMD)  $\gamma_k(F_P,F_Q)$ as the distance function. In the next section, we discuss, how to use dependent samples to estimate $\gamma_k(F_P,F_Q)$ and the corresponding error bounds induced by the dependence. For the remainder of the paper, we assume that the kernel function $k$ is characteristic such that $\gamma_k$ is a metric for probability measures.

\subsection{MMD Estimation for Dependent Observations}

% The unbiased and biased estimators for MMD with corresponding error bounds are provided in \cite{gretton2012kernel} under the i.i.d. assumption. The unbiased estimator can center around the true MMD with finite samples. However, the same estimator loses its unbiasedness with finite dependent samples. 
An error bound is presented in \cite{cherief2022finite} for the empirical mean estimator with samples satisfying the weak dependence condition in Assumption \ref{asp: weak dependence 1}. We adopt their results for the MMD estimator of $\gamma_k(F_P,F_Q)$ and the estimation error bounds under both null and alternative hypotheses. 

Let $\{X_i\}_{i=1}^{n+1}$ and $\{Y_i\}_{i=1}^{m+1}$ be the two independent sets of samples generated from kernels $P$ and $Q$, respectively. We have $\{\tilde X_i\}_{i=1}^n$ and $\{\tilde Y_i\}_{i=1}^m $ as the second-order Markov chains with kernels $\tilde P \coloneqq P^{\otimes 2}$ and $\tilde Q \coloneqq Q^{\otimes 2}$. Assume all samples are obtained from the stabilized Markov chains, i.e., $\tilde X_i \sim F_{P}$ and $\tilde Y_i \sim F_{Q}$. Then $\gamma_k(\hat F_{P_n}, \hat F_{Q_n})$ is an estimator of $\gamma_k(F_P,F_Q)$, where $\hat F_{P_n}=(1/n)\sum_{i=1}^n \delta_{\tilde X_i}$ is the empirical measure of $F_P$ and $\delta_{\tilde X_i}$ is the Dirac measure of sample $\tilde X_i$. Similarly, $\hat F_{Q_n}$ is defined for $F_Q$. Now, we introduce the weak dependence condition proposed in \cite{cherief2022finite}, which is necessary for the consistency of MMD estimation.
\begin{definition}
\label{def: rho coefficient}
    Define for any $t\in\mathbb{N}$,
    \begin{align}
    \label{eqn: weak dependence coefficient}
        \rho_t = \left|\mathbb{E}\langle k(X_t, \cdot)-\pi^0k, k(X_0, \cdot) - \pi^0k\rangle_{\ALP H}\right|.
    \end{align}
\end{definition}
\begin{assumption}
\label{asp: weak dependence 1}
    There exists a $\Sigma < +\infty$ such that $\sum_{t=1}^n \rho_t \leq \Sigma$ for $\forall n\in\mathbb{N}$, where $n$ is the length of the sample trajectory.
\end{assumption}
Note that Assumption \ref{asp: weak dependence 1} is automatically satisfied for the Markov chain $\{X_t\}_t$ satisfying the Doeblin's condition (Assumption \ref{asp: Doeblin condition}). As suggested by Proposition 4.4 in \cite{cherief2022finite}, Assumption \ref{asp: weak dependence 1} is more general than $\beta$-mixing property, which is a consequence of the Markov chain being uniformly ergodic due to Theorem 3.7 in \cite{bradley2005basic}. As a result of Lemma \ref{lem: UE for extended space}, the second-order chain $\{\tilde X_i\}_{i\in\mathbb{N}}$ also satisfies Assumption \ref{asp: weak dependence 1}.

% Assume that the kernel function can be written as $k(x,y)=F(||x-y||)$, where $F(a)=\int_a^\infty f(x)dx,$ for some non-negative continuous function $f$ such that $\int_0^\infty f(x)dx=1$. Due to Proposition 4.4 in \cite{cherief2022finite}, the coefficient defined in Definition \ref{def: rho coefficient} is bounded by $\rho_t\leq4(1-\lambda)^{\frac{t}{l}-1}$. For a Gaussian kernel $k$, the above condition is satisfied with the choice of $f(x) = \exp(-x/\sigma^2)/\sigma^2$. Therefore, Assumption \ref{asp: weak dependence 1} is satisfied with
% \begin{equation*}
% \label{eqn: Doeblin to Sigma}
%     \Sigma=\frac{4}{(1-\lambda)(1-(1-\lambda)^{1/l})}.
% \end{equation*}

Now, given Assumption \ref{asp: weak dependence 1} is satisfied, the MMD estimator using empirical distributions is asymptotically consistent for dependent observations, as we show below.

\begin{lemma} \label{lem: MMD consistency} (Consistency)
    Let $X=\{X_i\}_{i=1}^n$ be a Markov chain generated from Markov kernel $P$, which satisfies Assumption \ref{asp: weak dependence 1} with $\Sigma_X$. Let $\hat{P}_n=(1/n)\sum_{i=1}^n \delta_{X_i}$ is the empirical measure. Similarly, $Y = \{Y_j\}_{j=1}^n$ is generated using Markov kernel $Q$. We define  $\Sigma_Y, \hat{Q}_n$ analogously as above. Given that $k$ is characteristic, we have
    \begin{align*}
        \Big|\mathbb{E}\Big[\gamma_k\left(\hat{P}_n,\hat{Q}_n\right)\Big]-\gamma_k&(P,Q)\Big|\leq c_{X,Y}(n),
    \end{align*}
    where $c_{X,Y}(n)=\sqrt{\frac{1+2\Sigma_X}{n}}+\sqrt{\frac{1+2\Sigma_Y}{n}}.$
\end{lemma}
\begin{proof}
This is a direct consequence of Lemma 7.1 in \cite{cherief2022finite}.
\end{proof}

\subsection{Motivating Examples}
\label{sec: example section}
We now present examples of Markov chains induced by a linear time-invariant (LTI) system with random disturbance as input which satisfies both Assumptions \ref{asp: Doeblin condition} and \ref{asp: weak dependence 1}. 
\begin{example}(LTI system)
\label{ex: LTI}
    Consider a stable and controllable LTI system $(A, B)$ as following 
    $$x_{n+1} = Ax_{n} + Bw_{n}$$
    where system state $x_n\in\mathbb{R}^{d_1}$, control input $w_n\in\mathbb{R}^{d_2}$, and system dynamics $A\in\mathbb{R}^{d_1\times d_1}$, $B\in\mathbb{R}^{d_1\times d_2}$, for all $n \in \mathbb{N}$. Let the initial state be bounded such that $\|x_0\| \leq M_0$ for some positive finite $M_0$. The disturbance $(w_n)_{n\in\mathbb{N}}$ is a sequence of i.i.d. zero-mean random variable drawn from a distribution $\Gamma$ with compact support, i.e., there exists $M_w \in (0, \infty)$ such that $\|w_n\| \leq M_w$ for all $n\in\mathbb{N}$ almost surely. $\Gamma$ is non-singular with respect to the Lebesgue measure on $\mathbb R^{d_2}$. Let $\sigma_{\max}(\cdot)$ denote the largest singular value of a matrix. Then, we can conclude that for all $n\in\mathbb{N}$, 
    $$\|x_n\| \leq M_0 + \frac{\sigma_{\max}(A)}{1-\sigma_{\max}(A)}M_w,$$
    hence the state space of this Markov chain is compact. By Proposition 6.3.5 and Theorem 16.2.5 from \cite{meyn2012markov}, $\{x_n\}_{n\in\mathbb{N}}$ is an irreducible, aperiodic, T-chain with compact state space. Therefore, it is uniformly ergodic.
    It can be shown \cite[Proposition 4.5]{cherief2022finite} that given the $L$-Lipschitz function $F: \mathbb{R} \to \mathbb{C}$ such that the reproducing kernel $k(x,y)=F(\|x-y\|)$, the Assumption \ref{asp: weak dependence 1} is satisfied with 
    \begin{equation*}
        \Sigma=\frac{2L\sigma_{\max}(A)\sigma_{\max}(B)}{(1-\sigma_{\max}(A))^2}\mathbb{E}[||w_n||].
    \end{equation*}
    % For Gaussian kernel $k$ as defined in \eqref{ex: Gaussian kernel}, we have $L=1/(\sigma \sqrt{e})$. Further, 
    % given Gaussian distributed $w_n$, $\Sigma$ can be upper bounded using Jensen's inequality
    % \begin{align*}
    %     \mathbb{E}[||w_n||]\leq\sqrt{\mathbb{E}[||w_n||^2]}=\sqrt{tr(\Sigma_w)}.
    % \end{align*}
    % Further, if we assume bounded $v_n$ such that $||v_n||\leq v_{\max}$ almost surely for some constant $v_{\max}$, then Assumption \ref{asp: weak dependence 2} is satisfied with
    % \begin{equation}
    %     \Gamma=\frac{2v_{\max}\sqrt{L\lambda_{\max}}}{(1-\lambda_{\max})(1-\sqrt{\lambda_{\max}})}.
    % \end{equation}
\end{example}
A numerical simulation is provided in section \ref{sec: numerical simulations} based on this example. Due to limited space, we briefly characterize a uniformly ergodic nonlinear system. 
\begin{example}(Nonlinear system)
    For $t\geq 0$, $X_t$ and $W_t$ are random variables on $\mathbb{R}^n$ and $\mathbb{R}^m$, satisfying the following update equation.
    $$X_{t+1} = F(X_t, W_{t+1}),$$
    for some $C^\infty$ function $F: \mathcal{X} \times O \to \mathcal{X}$, where $\mathcal{X}$ is an open subset of $\mathbb{R}^n$ and $O$ is an open subset of $\mathbb{R}^m$. The random variables $\{W_t\}_{t\geq 0}$ are a sequence of noise with marginal distribution $\Gamma$ and density function $\gamma$ supported on an open set $O$. If the following conditions are satisfied, then $\mathbf{X} = \{X_t\}_{t\geq 0}$ is a uniformly ergodic Markov chain.
    \begin{enumerate}
        \item The nonlinear model with the following trajectories 
        \begin{align*}
            x_t = F_t(x_0, u_1, u_2, ..., u_t), \quad t\geq 1,
        \end{align*}
        is forward accessible if the set of the reachable state $A(x)\coloneqq \cup_{t=0}^\infty A_t(x)$ 
        % \begin{align}
        %     A(x)\coloneqq \bigcup_{t=0}^\infty A_t(x),
        % \end{align}
        has a non-empty interior for every $x\in\mathcal{X}$, where $A_t(x) \coloneqq \big\{F_t(x, u_1, ..., u_t):u_i\in O,1\leq i\leq t \big\}$ is the $t$ step reachable set from $x$. \label{cnd: 1}
        
        \item The density function $\gamma$ is lower semicontinuous on $\mathbb{R}^m$ and the support $O\coloneqq \{w\in\mathbb{R}^m:\gamma(w) > 0\}$ is open. \label{cnd: 2}
        
        \item Define $\Omega(x) \coloneqq \cap_{N=1}^{\infty}\{\cup_{t=N}^\infty A_t(x)\}$. There exists a global attracting state $x^* \in \Omega(y)$ for all $y \in \mathcal{X}$. \label{cnd: 3}
       
        \item The sets $O$ and $\Omega(x^*)$ are connected.  \label{cnd: 4}
        
        \item $\Omega(x^*)$ is compact. \label{cnd: 5}
    \end{enumerate}
    Conditions \ref{cnd: 1} and \ref{cnd: 2} yield that $\mathcal{X}$ is a T-chain (\cite[Proposition 7.1.2]{meyn2012markov}). Condition \ref{cnd: 3} allows the system to have a unique minimal invariant set $\Omega(x^*)$ when combined with \ref{cnd: 1} and \ref{cnd: 2} gives the fact that $\mathcal{X}$ is irreducible (\cite[Theorem 7.2.6]{meyn2012markov}). Condition \ref{cnd: 4} requires the system to be aperiodic (\cite[Proposition 7.3.4]{meyn2012markov}). Together with condition \ref{cnd: 5}, the Markov chain $\mathcal{X}$ satisfies the assumptions of Theorem 16.2.17 of \cite{meyn2012markov}, which implies the restriction of $\mathcal{X}$ on $\Omega(x^*)$ is uniformly ergodic.
\end{example}

% \begin{example}[Laplacian Kernel]
% For $x, y \in \mathbb{R}$ and $\sigma > 0$,
%     \begin{align*}
%         \psi(x-y) &= \exp(-\sigma\|x-y\|)
%     \end{align*}
% \end{example}
% \begin{example}[Inverse Multiquadratic Kernel]
% For $x, y \in \mathbb{R}$, $c > 0$, and $\sigma > 0$,
%     \begin{align*}
%         \psi(x-y) &= 1/(\sigma^2+\|x-y\|^2)^c
%     \end{align*}
% \end{example}

% Now, we introduce the weak dependence coefficient defined with the kernel mean embedding and MMD\cite{cherief2022finite}. 
% \begin{definition}
% \label{def: rho coefficient}
%     Define for any $t\in\mathbb{N}$,
%     \begin{align}
%     \label{eqn: weak dependence coefficient}
%         \rho_t = \left|\mathbb{E}\langle k(X_t, \cdot)-\pi^0k, k(X_0, \cdot) - \pi^0k\rangle_{\ALP H}\right|.
%     \end{align}
% \end{definition}
% We introduce the following weak dependence condition which is imposed on the observation in the performance analysis. This is considered a more general condition than $\beta$-mixing \cite{cherief2022finite}.

We next present the test statistic and analyze its performance.

\section{Main Results}
\label{sec: main results}

In this section, we introduce a CUSUM-typed online test statistic based on the MMD norm, and we give the theoretical guarantees on MTBFA and MD when applied to Markovian observations. The Kernel CUSUM algorithm is proposed for independent samples in \cite{flynn2019change}. Our test statistics and theirs are similar in form. However, it requires a completely different approach to derive the theoretical bounds for MTBFA and MD, as the samples are dependent. 

We first assume the availability of previous system observations under the pre-change kernel. The length of the observed sample path is sufficiently long. Let $X^{m+1} = \{X_t\}_{t=1}^{m+1}$ denote past $m+1$ observations and $\tilde X_t = (X_t, X_{t+1})$. Let $Y_t$ be the current system observation at time $t \geq 0$, whose kernel is subject to change from $P$ to $Q$, and $\tilde Y_t=(Y_t,Y_{t+1})$.
% generated by the pre-change Markov kernel $P$. We denote the concatenation of system states as $Y_t = (X_t, X_{t-1})$ and $Z_t = (Y_t, Y_{t-1})$. 
% Notice that, by the argument in the proof of Lemma 1 in \cite{xian2016online}, $\{Y_t\}$ and $\{Z_t\}$ are uniform ergodic Markov chains with kernels $P^{\otimes 2}$, $P^{\otimes 3}$ and invariant distribution $\pi_{P}\otimes P$, $\pi_P \otimes P^{\otimes 2}$, respectively. 
To estimate MMD from samples, we should maintain a buffer $B_t^r = \{Y_i\}_{i=t-r}^{t}$ of size $r\in\mathbb{N}$. At time $t$, the buffer always contains the most recent $r+1$ observation so that we can have $r$ concatenated states $\{\tilde Y_i\}_{i=t-r+1}^{t}$.

Now, let us introduce the test statistics and the stopping rule by first defining the following function $s: \mathbb{X}^r \to \mathbb{R}$,
\begin{align} \label{eqn: MMD estimator with r and m}
    & s(B_t^r) = \Bigg(\frac{1}{r^2}\sum_{t-r\leq i,j\leq t} k(\tilde Y_i, \tilde Y_j) + \frac{1}{m^2}\sum_{1\leq i,j \leq m} k(\tilde X_i, \tilde X_j) \nonumber\\
    & - \frac{2}{mr}\sum_{i=t-r}^{t}\sum_{j=1}^{m}k(\tilde Y_i, \tilde X_j)\Bigg)^{1/2} - c_{\tilde X, \tilde Y}(m, r) - \epsilon,
\end{align}
where $c_{\tilde X, \tilde Y}(m, r) = \sqrt{\frac{1 + \Sigma_{\tilde X}}{m}} + \sqrt{\frac{1 + \Sigma_{\tilde Y}}{r}}$ is the estimation error bound defined in Lemma \ref{lem: MMD consistency}, $\epsilon > 0$ is a small arbitrary constant, and $\Sigma_{\tilde X}$, $\Sigma_{\tilde Y}$ are defined in Assumption \ref{asp: weak dependence 1}. To justify the use of $\epsilon$, we argue that $c_{\tilde X, \tilde Y}(m, r)$ cannot be computed exactly. An overestimation of $c_{\tilde X, \tilde Y}(m, r)$ is often preferred over an underestimation to ensure a low false alarm rate. Note that when the past observation and buffer have fixed length $m$ and $r$, $c_{\tilde X,\tilde Y}(m, r)$ is a constant. Let the cumulative sum of the function $s$ from $t=k$ to $n$ be $S_{k:n}$, where $k \geq 1$ and $n \geq k+1$,
\begin{align}
\label{eqn: cusum of mmd}
   S_{k:k} = s(B^r_k),\quad S_{k:n} = S_{k:n-1} + s(B^r_n).
\end{align}
Define the test statistics $\hat S_{n}$ and the corresponding stopping rule with threshold $b \geq 0$ and minimum sample $M\in\mathbb{N}$ as:
\begin{align} \label{eqn: test statistic}
    \hat S_{n} &= \max_{1\leq k \leq n-M} S_{k:n},
\end{align}
\begin{equation} \label{eqn: stopping rule}
    T(b, M) = \inf\left\{n\in\mathbb{N}:\ \hat S_{n} \geq b,\ n > M\right\}.
\end{equation}
% Taking maximum in Equation \ref{eqn: cusum of mmd} makes $S_{k:n} \geq 0$ for all $1\leq k\leq n$, hence the stopping rule has a simplified version which is easier to analyze. 
% \begin{align}
%     T(b, M) = \inf\left\{n\in\mathbb{N}:\quad S_{1:n} \geq b\right\}.
% \end{align}
% Note that the test statistics and stopping rule are dependent on the size of past observation $O^m+1$ and the size of the buffer $B_t^r$. The effect of $m$, $r$ on our results will be discussed in later sections, however, for simplicity, we omitted them from current notion as they are fixed in the following analysis.

% \subsection{Mean Time Between False Alarms (MTBFA)}
% Recall the mean time between false alarm is the testing error under null hypothesis, i.e. the system evolve according to Markov kernel $P$ only throughout the period of observation. We denote the expectation and the probability law taken under this assumption as $\mathbb{E}_{\infty}[\cdot]$ and $\mathbb{P}_{\infty}\{\cdot\}$, respectively. Then, $MTBFA(T(b, M))$ is written as $\mathbb{E}_\infty[T(b, M)]$, and 

% Let $P$ and $Q$ be the pre- and post-change Markov kernel. Let $(O_i)_{i=1}^m$ be generated from $P$, and $(X_i)_{i=1}^t$ be generated from $Q$. 
Let $T$ be the stopping time as defined in \eqref{eqn: stopping rule}. For fixed sample size $m$ under pre-change kernel $P$ and a sliding window of size $r$ under post-change kernel $Q$, we have the following results.
\begin{theorem}[MTBFA Lower Bound]
\label{thm: MTBFA}
Suppose that Markov kernel $P$ satisfies Doeblin's condition and $Q=P$. Denote $\tilde P = P^{\otimes r+1}$. There exists $\alpha >0$ such that for $b$ sufficiently large and $\epsilon > 0$, we have
    \begin{align*}
        \text{MTBFA}(T(b, M))\geq M+\frac{\sqrt{2}-2}{6} + \frac{2}{3}\exp\bigg[\frac{4\epsilon(b - \alpha)}{\alpha^2}\bigg].
    \end{align*}
\end{theorem}

\begin{proof}
    Please refer to Appendix \ref{app: MTBFA}.
\end{proof}

\begin{theorem}[MD Upper Bound]
\label{thm: MD}
    Suppose that $P$ is the pre-change kernel and $Q$ is the post-change kernel, and they both satisfy Doeblin's condition. Assume that MMD kernel $k$ is characteristic and $\|k\|_\infty\leq 1$. Then, there exists $\alpha>0$ such that for sufficiently small $\epsilon$,
    \begin{equation*}
        \text{MD}(T(b,M))\leq\max\left\{M,\frac{b+\alpha}{D_r(P,Q)-\epsilon}\right\}(1+o(1)),
    \end{equation*}
    where $D_r(P,Q)$ is defined before \eqref{eqn: Dr(P,Q)}.
    % $$D_r(P,Q)=\gamma_k(P,Q)-2\left(\sqrt{\frac{1+2\Sigma_P}{r}}+\sqrt{\frac{1+2\Sigma_Q}{m}}\right).$$
\end{theorem}
\begin{proof}
    Please refer to Appendix \ref{app: MD}.
\end{proof}

\begin{figure}[h]
        \centering
        \includegraphics[width=0.8\linewidth]{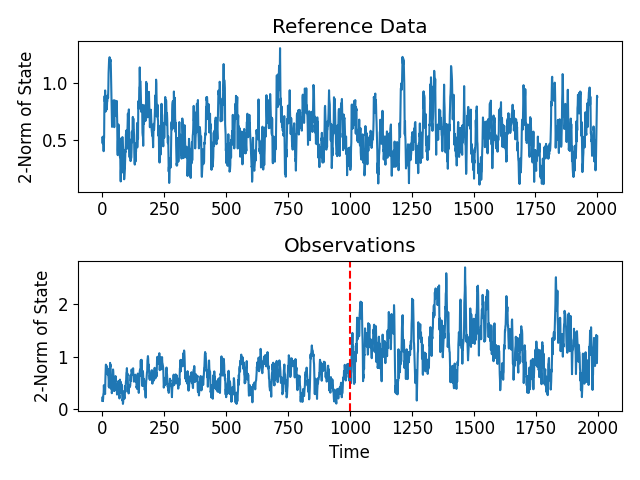}
        \includegraphics[width=0.8\linewidth]{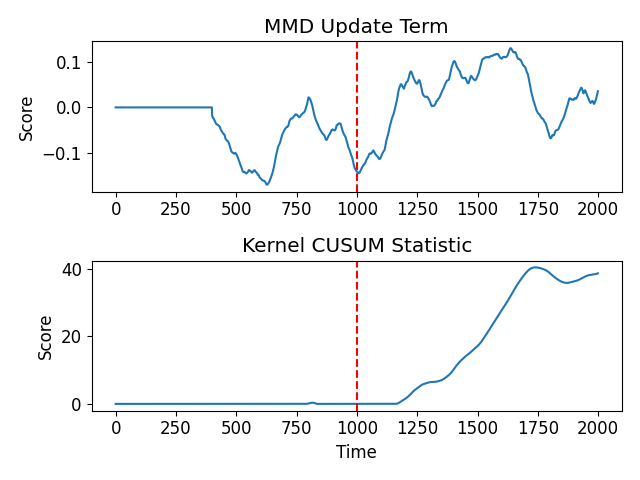}
        \caption{A change in the variance of $\omega_n$'s distribution is applied at the red breaking line ($\tau=1000$). The left plot shows the 2-norm of the system state. The middle plot shows the update term $s(B_t^r)$ at each step. The right plot shows the kernel CUSUM statistic $\hat S_n$.}
        \label{fig: mean change}
\end{figure}
\begin{figure}[h]
    \centering
    \includegraphics[width=0.8\linewidth]{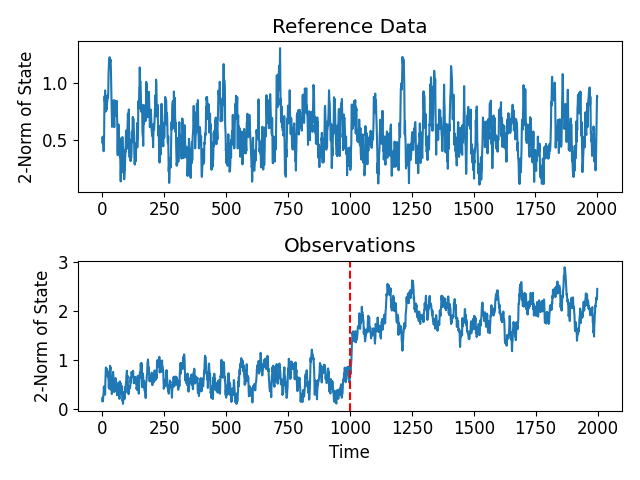}
    \includegraphics[width=0.8\linewidth]{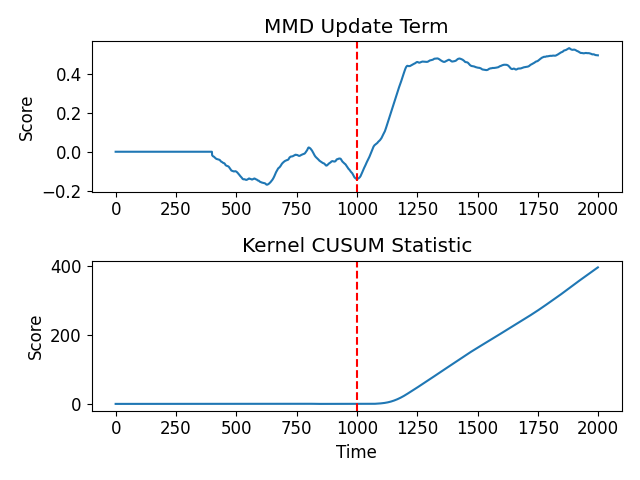}
    \caption{A change in the mean of $\omega_n$'s distribution is applied at the red breaking line ($\tau=1000$).}
    \label{fig: var change}
\end{figure}

\section{Numerical Simulations}
\label{sec: numerical simulations}
In this section, we evaluate the performance of the kernel CUSUM change detector through numerical simulations. We use auto-regressive (AR) processes with an abrupt change in the noise statistics to mimic the Markov chain with a change in the transition kernel. Let $A$ be a square matrix with a spectral radius less than 1, and $\{\omega_n\}_{n\in\Na}$ be a sequence of i.i.d. samples drawn from some distribution. Two scenarios of changes are simulated in the experiment: 1) a change in the variance of $\omega_n$; 2) a change in the mean of $\omega_n$. The state of the AR process $X_n$ evolves as %3) change in distribution type of $\omega_n$.
\begin{align}
\label{eqn: AR process}
    X_{n+1} = AX_n + \omega_n
\end{align}
We use the average of three Gaussian kernels with $\sigma = 10^{-1}, 1, 10$ to capture values of different scales. We set $X_n\in\mathbb{R}^4$ throughout all simulations. The system matrix $A$ is chosen with a spectral radius of $0.95$. Pre-change system observations (reference data) are generated with $\omega_n$ i.i.d. drawn from a zero mean truncated Gaussian distribution with an identity covariance matrix multiplied by 0.1. The following modification was done to the distribution of $\omega_n$.
\begin{enumerate}
    \item Change in variance: after the change point, the covariance becomes an identity matrix multiplied by 0.2, while the mean of the Gaussian vector stays zero;
    \item Change in mean: after the change point, the mean of the Gaussian vector becomes 0.05, while the covariance stays the same as the pre-change distribution;
\end{enumerate}
Before each experiment, we calibrate the parameter $c_{\tilde X, \tilde Y}(m, r) + \epsilon$ by subtracting an appropriate constant from the MMD update term $s(B_t^r)$ such that it is negative at each time step when there is no change. 

Figures \ref{fig: mean change} and \ref{fig: var change} show that the kernel CUSUM change detector detects the change as expected. In both scenarios, the MMD update term turns positive, and the test statistic grows shortly after the change point. Overall, the proposed kernel CUSUM change detector for Markov chains can detect the change when the problem lacks distributional assumption for both pre- and post-change Markov kernel. 

% \section{Discussion}
% \label{sec: discussion}

\section{Conclusion}
\label{sec: conclusion}
In this study, we proposed a MMD-based kernel CUSUM change detector for Markov chains. 
%The advantages of our method include: 1) being completely data-driven with no assumption on both pre- and post-change Markov kernels; 2) working well with high dimensional data and compatible with general state space Markov chains. 
For the proposed algorithm, we derived a lower bound on the mean time between false alarm (MTBFA) and an upper bound on the mean delay (MD). We further present two numerical simulations of change detection in autoregressive processes to demonstrate the effectiveness of the proposed method. 

For future works, we will investigate the relaxation of the uniform ergodicity condition in Assumption \ref{asp: Doeblin condition}. 
%Under a slightly more technical weak dependence condition as in Assumption 3.1 from \cite{cherief2022finite} and Assumption \ref{asp: weak dependence 1}, there exists a new version of McDiarmid's inequality \cite{rio2013mcdiarmid} which could provide an alternative approach to bound MTBFA and MD. 
The application of this change detection method to attack detection and dynamic watermarking of general nonlinear systems will also be an objective of future works. 

% It would be of interest to find the weakest condition on a Markov chain such that Assumption \ref{asp: weak dependence 1} and Assumption 3.1 from \cite{cherief2022finite} can be satisfied. 

\section{Appendix}
\subsection{Proof of Lemma \ref{lem: RKHS for MC}}
\label{app: RKHS for MC}
There can be two cases: $\pi_P\neq\pi_Q$ and $\pi_P = \pi_Q$. In case $\pi_P\neq\pi_Q$, then it is obvious that $\pi_P \otimes P \neq \pi_Q\otimes Q$ since the marginal of the two measures are different. In the other case, if $\pi_P = \pi_Q$, then pick a closed set $\mathbb A\subset\mathbb{X}$ such that $P(x,\cdot) \neq Q(x,\cdot)$ for all $x\in\mathbb A$ and $\pi_P(\mathbb A)>0$. Let $\mathbb B_x, x\in\mathbb A$ be any closed set such that $P(x,\mathbb B_x) >Q(x,\mathbb B_x)$, which exists due to the Jordan-Hahn decomposition theorem. Consider the set $\mathbb C = \cup_{x\in\mathbb A}\{x\}\times\mathbb B_x$ and note that $\overline{\mathbb C}$ is a measurable set since it is a closed set. 

We claim that any section $\overline{\mathbb C}_x:=\{y:(x,y)\in\overline{\mathbb C}\}$ is equal to $\mathbb B_x$ for all $x\in\mathbb A$. Contrary to the claim, suppose that $\overline{\mathbb C}_{\bar x}\neq \mathbb B_{\bar x}$ for some $\bar x\in\mathbb A$ and pick a point $\bar y\in\overline{\mathbb C}_{\bar x}\setminus \mathbb B_{\bar x}$. Pick a sequence $\{(\bar x,\bar y_n)\}_{n\in\Na}\subset \mathbb C$, such that $(\bar x,\bar y_n)\to(\bar x,\bar y)\in \overline{\mathbb C}$. This implies that $\{\bar y_n\}_{n\in\Na}\subset \mathbb B_{\bar x}$ and $\bar y_n\to \bar y$. Since, $\mathbb B_{\bar x}$ is a closed set, we have $\bar y\in\mathbb B_x$, which contradicts our assumption that $\overline{\mathbb C}_{\bar x}\neq \mathbb B_{\bar x}$. Thus, $\overline{\mathbb C}_x= \mathbb B_x$ for all $x\in\mathbb A$.

From Proposition 3.3.2 in \cite{bogachev2007measure}, we conclude that the functions $x\mapsto P(x,\mathbb B_x)$ and $x\mapsto Q(x,\mathbb B_x)$ are measurable functions since they are the composition of measurable functions. Thus, we conclude that
\beqq{&\pi_P \otimes P(\overline{\mathbb C}) = \int_{\mathbb A} P(x,\mathbb{B}_x) \pi_P(dx) \\
& > \int_{\mathbb A} Q(x,\mathbb{B}_x) \pi_P(dx) = \pi_Q\otimes Q(\overline{\mathbb C}) = 0.}
Thus, $\pi_P \otimes P \neq \pi_Q\otimes Q$ and the proof is complete.

\subsection{Proof of Lemma \ref{lem: UE for extended space}}
\label{app: UE for extended space}
Given $(X_t)_{t\in\mathbb{N}}$ is uniformly ergodic with $(\Omega,\mathcal{F},\mathbb P)$, there exists a $V_X(x)\geq1$ such that 
$$PV_X(x)-V_X(x)\leq-\beta V_X(x)+b1_C(x),$$
for some $\beta>0$, $b<\infty$, and some petite set $C$. As defined in Section 5.2 of \cite{meyn2012markov}, there exists a $m>0$, and a non-trivial measure $\nu_m$ on $\mathcal{B}(\Omega)$ such that for all $x\in\mathcal{C}$, $A\in\mathcal{B}(\Omega)$, $P^m(x,A)\geq\nu_m(A)$. Now for the second-order system $\tilde X_t=(X_t,X_{t+1})$, let $V_{\tilde X}(\tilde X_t)=V_X(X_{t+1})\geq 1$, we have
\begin{align*}
    & P^{\otimes 2}V_{\tilde X}(\tilde X_t)-V_{\tilde X}(\tilde X_t)=PV_X(X_{t+1})-V_X(X_{t+1})\\
    & \leq-(1+\beta) V_X(X_{t+1})+b1_C(X_{t+1})\\
    & \leq-(1+\beta) V_{\tilde X}(\tilde X_t)+b1_{C_0\times C}(\tilde X_t).
\end{align*}
for some $C_0\subseteq\mathcal{B}(\Omega)$ such that $P(C_0,C)>0$. Next, we need to show the petiteness of $C_0\times C$ in the second-order system. Let $\tilde\nu_{m+1}(C_1,C_2)=P(C_1,C_2)\nu_{m}(C_2)$, then for all $(x_1,x_2)\in C_0\times C_1$, and $(A_1,A_2)\in\mathcal{B}(\Omega\times\Omega)$ we have
\begin{align*}
    & P^{\otimes 2,m+1}((x_1,x_2),(A_1,A_2))\\ & =P^m(x_2,A_1)P(A_1,A_2)\geq\nu_m(A_1)P(A_1,A_2)\\
    & =\tilde\nu_{m+1}(A_1,A_2).
\end{align*}
Therefore by the same drift condition for the second-order system, $(\tilde X_t)_{t\in\mathbb{N}}$ satisfies Doeblin's condition.

\subsection{Proof of Theorem \ref{thm: MTBFA}}
\label{app: MTBFA}
The proof of MTBFA bound is inspired by \cite{xian2016online} with key steps replaced by the following Hoeffding's inequality for uniformly ergodic Markov chain on general state space. 
\begin{theorem}[Hoeffding's Inequality \cite{glynn2002hoeffding}]
\label{thm: Hoeffding inequality}
    Suppose that the Markov chain $\{X_t, t\geq 0\}$ is uniformly ergodic and let function $f: \mathbb{X} \to \mathbb{R}$ with $\|f\| = \sup_{x\in\mathbb{X}} f(x) < \infty$. Let $S_n = \sum_{t=1}^n f(X_t)$, $\mu = 2(l+1)\|f\|/\lambda$, where $l$ and $\lambda$ are described in Assumption \ref{asp: Doeblin condition}. Then, for any $\epsilon > 0$ and $n > \mu/\epsilon$, we have
    \begin{equation*}
        \mathbb{P}(\|S_n - \mathbb{E}[S_n]\| \geq n\epsilon) \leq 2 \exp\left\{-2\frac{(n\epsilon - \mu)^2}{n\mu^2}\right\}.
    \end{equation*}
\end{theorem}

For threshold $b \geq 0$, minimum sample $M$, and stopping rule $T(b, M)$, we have,
\begin{align}
    \mathbb{E}_\infty [T(b, M)] &= \sum_{C=1}^\infty \mathbb{P}_{\infty}\{T(b, M) > C\}\nn\\
    & = M + \sum_{C=M+1}^\infty \mathbb{P}_\infty\{T(b, M) > C\}. \label{eqn: mtbfa eqn 1}
\end{align}
Due to Lemma \ref{lem: UE for extended space}, $(B_t^r)_{t \geq r}$ is also a uniformly ergodic Markov chain with transition kernel $\tilde P \triangleq P^{\otimes r+1}$ and invariant measure $\pi_P\otimes P^{\otimes r}$. For $t > r$, $s(B^r_t)$ is a biased estimation of $\gamma_k(\pi_P\otimes P^{\otimes r-1}, \pi_P\otimes P^{\otimes r-1}) = 0$ minus the error bound in Lemma \ref{lem: MMD consistency}. Thus, we have $-\epsilon \geq \mathbb{E}[s(B_t^r)]\geq-2c_{\tilde X, \tilde Y}(r,m)$. Recall that $\epsilon$ is a small positive constant picked by the user. Let $l_{\tilde P}$ and $\lambda_{\tilde P}$ be the coefficients in Assumption \ref{asp: Doeblin condition} satisfied by Markov kernel $\tilde P$. Let $\alpha := 2(l_{\tilde P}+1)/\lambda_{\tilde P}$. Apply Theorem \ref{thm: Hoeffding inequality} to the following probability when $n > \max\{(\alpha-b), 0\}/\epsilon$,
\begin{align}
\label{eqn: mtbfa eqn 2}
    &\mathbb{P}_\infty(S_{1:n}>b)=\mathbb{P}_\infty\left(\sum_{t=0}^{n}s(B^r_t)>b\right)\nonumber\\
    =&\mathbb{P}\left(\sum_{t=0}^{n}s(B^r_t)-n\mathbb{E}_\infty[s(B_t^r)]>b-n\mathbb{E}_\infty[s(B_t^r)]\right)\nonumber\\
    \leq&\exp\left(-2\frac{(b - n\mathbb{E}_\infty[s(B_t^r)]-\alpha)^2}{n\alpha^2}\right)\\
    \leq&\exp\left(-2\frac{(b - \alpha +\epsilon n)^2}{n\alpha^2}\right).
\end{align}
The the right-hand side (RHS) attains maximum when $n = (b-\alpha) / \epsilon$. Expanding the RHS of \ref{eqn: mtbfa eqn 1} and applying the inequality in \ref{eqn: mtbfa eqn 2}. After rearranging the terms, we arrive at
\begin{align}
    &\sum_{C=M+1}^\infty \mathbb{P}_\infty \{T(b, M) > C\} \nn \\
    %= \sum_{C=M+1}^\infty (1 - \mathbb{P}_{\infty}\{T(b, M)\leq C\})\nonumber\\
    =& \sum_{C=M+1}^\infty \left(1 - \mathbb{P}_\infty\left\{\bigcup_{n=M+1}^{C} \{T(b, M) = n\}\right\}\right)\nonumber\\
    \geq& \sum_{C=M+1}^L \left(1 - \mathbb{P}_\infty\left\{\bigcup_{n=M+1}^{C}\bigcup_{k=1}^{n-M}\{S_{k:n} \geq b\}\right\}\right)\nonumber\\
    \geq& \sum_{C=M+1}^L \left(1 - \sum_{n=M+1}^{C}\sum_{k=1}^{n-M}\mathbb{P}_\infty\{S_{k:n} \geq b\}\right)\nonumber\\
    =& \sum_{C=1}^{L-M}\left(1 - \sum_{l=1}^{C}\sum_{k=1}^{l}\mathbb{P}_\infty\{S_{k:l+M} \geq b\}\right) \nonumber\\
    \geq& \sum_{C=1}^{L-M}\left\{1 - \frac{C(C+1)}{2}\exp{\left(-\frac{8\epsilon(b-\alpha)}{\alpha^2}\right)}\right\},\label{eqn: mtbfa eqn 3}
\end{align}
where $L > M$. When $L = L^*\coloneqq M + \frac{1}{2}\bigg(\sqrt{1+8\exp (\frac{8\epsilon(b-\alpha)}{\alpha^2})}-1\bigg)$, the RHS of the equation \eqref{eqn: mtbfa eqn 3} reaches maximum, where $L^*$ is the largest solution of 
\begin{align}
\label{eqn: original L eqn}
    (L-M)(L-M+1) = 2\exp\left[\frac{8\epsilon(b-\alpha)^2}{\alpha^2}\right].
\end{align}
$L^* \geq M+1$, whenever $b \geq \alpha$. Plug $L^*$ into \eqref{eqn: mtbfa eqn 3}, we have for large $b$,
\begin{align}
\label{eqn: mtbfa eqn 4}
    &\sum_{C=M+1}^\infty \mathbb{P}_\infty\{T(b, M) > C\}\nonumber\\
    \geq& L^*-M - \frac{1}{2}\exp{\left(-\frac{8\epsilon(b-\alpha)}{\alpha^2}\right)}\sum_{C=1}^{L^*-M}C(C+1)\nonumber\\
    \geq& \frac{2}{3}\exp\left(\frac{4\epsilon(b-\alpha)}{\alpha^2}\right) + \frac{\sqrt{2}-2}{6}.
\end{align}
% Recall that, $L^*$ is the solution of equation \eqref{eqn: original L eqn}, which cannot be solved analytically. To lower bound the last term in equation \eqref{eqn: mtbfa eqn 4}, we need to estimate $L^*$ from below. We solve the following alternative equation and denote the solution as $\hat L$,
% \begin{align}
% \label{eqn: surrogate eqn}
%     \hat L-M+1 &= \sqrt{2} \bigg [1 + \frac{(b-\alpha)^2}{\hat L\alpha^2}\bigg]\\
%     \hat L = \frac{1}{2}\bigg[M+\sqrt{2}-1 + &\sqrt{(M+\sqrt{2}-1)^2 + \frac{4\sqrt{2}(b-\alpha)^2}{\alpha^2}}\bigg].\nonumber
% \end{align}
% We claim $L^* \geq \hat L$ for reasons presented in the next paragraph. Plugging $\hat L$ in equation \eqref{eqn: mtbfa eqn 4}, we have
% \begin{align}
% \label{eqn: mtbfa approx lower bound}
%     &\sum_{C=M+1}^\infty \mathbb{P}_\infty\{T(b, M) > C\}\nonumber\\
%     \geq & \kappa\exp\bigg[(\frac{1}{\alpha} + \frac{2\epsilon}{\alpha^2} +\frac{\sqrt{2}\epsilon^2}{\alpha^3})(b-\alpha)-\frac{\sqrt{2}M}{8}\bigg],
% \end{align}
% where $\kappa$ is a positive constant. 
Combining the equations \eqref{eqn: mtbfa eqn 1} and \eqref{eqn: mtbfa eqn 4}, we arrive at the desired result. 

% Now, we explain why $L^* \geq \hat L$. Indeed, as $L > \max\{M, (\alpha - b)/\epsilon\}$, both sides of equation \eqref{eqn: original L eqn} are increasing in $L$. If we majorize the RHS (or minorize the LHS) and equate it with LHS (or RHS), the resulting root reduces. Thus, the equation \eqref{eqn: surrogate eqn} is obtained by: 1) removing the $\epsilon L$ in the exponent and shrinking the RHS; 2) majorizing LHS to $(L-M+1)^2$; 3) taking squared root on both sides; 4) lower-bounding RHS with its first order Taylor approximation. 

\subsection{Proof of Theorem \ref{thm: MD}}
\label{app: MD}
    Since $P$ and $Q$ satisfies Assumption \ref{asp: weak dependence 1}, so does the $\tilde P$ and $\tilde Q$ with weak dependent coefficient $\Sigma_{\tilde P}$ and $\Sigma_{\tilde Q}$, where
    \begin{align*}
        \Sigma_{\tilde P}=\frac{4}{(1-\lambda_P)(1-(1-\lambda_P)^{1/l_P})},
    \end{align*}
    and a similar result holds for $\Sigma_{\tilde Q}$. Let $D_r(P,Q):=\gamma_k(\tilde P,\tilde Q)-2c_{X,Y}(r,m)$.
    Note that for any $t>r$, $s(B_t^r)$ defined in \eqref{eqn: MMD estimator with r and m} is an estimator of the MMD between $\tilde P$ and $\tilde Q$, using $m$ samples from $\tilde P$ and $r$ samples from $\tilde Q$. By Lemma \ref{lem: MMD consistency}, $s(B_t^r)$ is consistent such that
    \begin{align}
    \label{eqn: Dr(P,Q)}
        \mathbb{E}[s(B_t^r)]\geq\gamma_k(\tilde P,\tilde Q)-2c_{X,Y}(r,m) - \epsilon = D_r(P,Q) - \epsilon,
    \end{align}
    where $c_{X,Y}(r,m)=\sqrt{\frac{1+2\Sigma_{\tilde P}}{r}}+\sqrt{\frac{1+2\Sigma_{\tilde Q}}{m}}$ and $\epsilon < D_r(P,Q)$ is picked sufficiently small. Assume the change point $\tau=0$, and for any $n>n_0=\max\left\{M,\frac{b+\alpha}{D_r(P,Q)-\epsilon}\right\}$, where $\alpha=2m_{\tilde Q}/\lambda_{\tilde Q}$, we have
    \begin{align*}
        &\mathbb{P}(S_{0:n}<b)=\mathbb{P}\left(\sum_{t=0}^{n}s(B^r_t)<b\right)\nonumber\\
        %=&\mathbb{P}\left(\sum_{t=0}^{n}s(B^r_t)-n\mathbb{E}[s(B_t^r)])<-(n\mathbb{E}[s(B_t^r)]-b)\right)\\
        \leq&\exp\left(-2\frac{(n\mathbb{E}[s(B_t^r)]-\alpha-b)^2}{n\alpha^2}\right)\\
        \leq&\exp\left(-2\frac{n(D_r(P,Q)-\epsilon)-\alpha-b)^2}{n\alpha^2}\right),
    \end{align*}
    where the first inequality follows from Theorem \ref{thm: Hoeffding inequality}. The rest of the proof follows from (14) and (15) in \cite{xian2016online}, where $I(Q,P)$ is replaced with $D_r(P,Q)-\epsilon$.

\bibliographystyle{ieeetr}
\bibliography{main}

\end{document}